\documentclass[%
 aip,
 jmp,%
 amsmath,amssymb,
preprint,
]{revtex4-2}

\usepackage[utf8]{inputenc}
\usepackage[T1]{fontenc}
\usepackage{mathptmx}
\usepackage{etoolbox}
\usepackage{amsthm}
\usepackage{bbm}
\usepackage{hyperref}
\usepackage{tikz-cd}

\newtheorem{theorem}{Theorem}
\newtheorem{definition}{Definition}
\newtheorem{lemma}[definition]{Lemma}
\newtheorem{remark}[definition]{Remark}

\makeatletter
\def\@email#1#2{%
 \endgroup
 \patchcmd{\titleblock@produce}
 {\frontmatter@RRAPformat}
 {\frontmatter@RRAPformat{\produce@RRAP{*#1\href{mailto:#2}{#2}}}\frontmatter@RRAPformat}
 {}{}
}%
\makeatother
\begin{document}

%\preprint{AIP/123-QED}

\title{From Kraus Operators to the Stinespring Form of Quantum Maps: An Alternative Construction for Infinite Dimensions}
% Force line breaks with \\
\author{Frederik vom Ende}
% \altaffiliation[Also at ]{Physics Department, XYZ University.}%Lines break automatically or can be forced with \\
%\author{B. Author}%
 \email{frederik.vom.ende@fu-berlin.de}
\affiliation{ 
Dahlem Center for Complex Quantum
Systems, Freie Universität Berlin, Arnimallee 14, 14195 Berlin, Germany}

\date{\today}% It is always \today, today,
 % but any date may be explicitly specified

\begin{abstract}
We present an alternative (constructive) proof of the statement that for every 
completely positive, trace-preserving map $\Phi$ there exists an auxiliary Hilbert 
space $\mathcal K$ in a pure state $|\psi\rangle\langle\psi|$ as well as a unitary 
operator $U$ on system plus environment such that
$\Phi$ equals $\operatorname{tr}_{\mathcal K}(U((\cdot)\otimes|\psi\rangle\langle\psi|)U^*)$.
The main tool of our proof is Sz.-Nagy's dilation theorem applied to isometries defined on a subspace.
In our construction, the environment consists of a system of dimension ``Kraus rank of $\Phi$'' together with a qubit, the latter only acting as a catalyst.
In contrast, the original proof of Hellwig \& Kraus given in the 70s yields an auxiliary system of dimension ``Kraus rank plus one''. We conclude by providing an example which illustrates how the constructions differ from each other.
\end{abstract}

% PACS, the Physics and Astronomy Classification Scheme.

%Use showkeys class option if keyword display desired

\maketitle
\section{Introduction}

It is well known that every completely positive trace-preserving map
can be represented by first coupling to an auxiliary system, followed by applying a global unitary operation, and finally discarding the auxiliary system, cf., e.g., Thm.~6.18 in \cite{Holevo12}.
This representation -- sometimes called Stinespring dilation or Stinespring form of quantum maps -- is a natural approach in the analysis of important properties of completely positive maps.
Due to its direct physical interpretation it can even be applied to certain experimental setups \cite{Braun01,Haake10}.

The Stinespring form is crucial to many areas of quantum physics. To name just a few: in quantum thermodynamics this is how the allowed operations in the corresponding resource theory approach
%(called ``thermal operations'') 
are defined \cite{Lostaglio19,vomEnde22thermal},
in quantum control it is used to emulate arbitrary 
%(possibly non-Markovian) 
quantum maps via restricting Markovian evolutions of larger systems \cite{PRL_decoh2,JPB_decoh},
and in continuous-variable quantum information theory \cite{Weedbrook12} 
the channels which have a Stinespring form (where the unitary on the full system is Gaussian)
have been shown to approximately coincide with the
set of linear bosonic channels \cite{Lami18}.
Indeed, the Stinespring form of bosonic Gaussian channels has been studied a fair bit in the past \cite{Holevo07,Caruso08,Caruso11}
as the main information of these infinite-dimensional systems
boils down to 
the underlying classical phase space (i.e.~something finite-dimensional).
This approach has found use in studying, e.g, the Holevo quantity \cite{Shirokov20} or
the black hole quantum information loss problem \cite{Bradler15,Kretschmann08}.

Motivated by the widespread importance of this representation we will take a step back and have a look at how one proves its existence in the first place.
Given a linear map $\Phi$ on $\mathbb C^{n\times n}$, $n\in\mathbb N$, if
\begin{equation}\label{eq:intro_0}
\Phi\equiv\operatorname{tr}_{\mathcal K}\big(U((\cdot)\otimes\omega)U^*\big)
\end{equation}
for some complex Hilbert space $\mathcal K$, some state $\omega$ on $\mathcal K$, and some unitary $U$ on $\mathbb C^n\otimes\mathcal K$,
then the r.h.s.~of \eqref{eq:intro_0} is called \textit{Stinespring form} of $\Phi$.
It is an established fact that a linear map has a Stinespring form if and only if it is completely positive
%\footnote{
%As usual, a linear map $\Phi$ is called \textit{completely positive} if $\Phi\otimes\operatorname{id}_k$ (or $\operatorname{id}_k\otimes\Phi$) for all $k\in\mathbb N$ maps positive semi-definite matrices to positive semi-definite matrices.}
and trace-preserving, and this characterization even extends to infinite dimensions \cite{HK70,Kraus71}.
To illustrate the common construction -- first in the simpler case of finite dimensions --
one starts with a set of Kraus operators $\{K_i\}_{i=1}^{\ell}$ of $\Phi$ and collects them ``in the first column'' of a larger matrix $U_0$.
Then one fills up the rest of $U_0$ such that it becomes unitary:
%More precisely \cite[Eq.~(6.23)]{Holevo12} define
%\begin{align*}
%U_0:\mathbb C^n\otimes \mathbb Ce_{1}=\{x\otimes e_{1}:x\in\mathcal H\}&\to\mathbb C^n\otimes\mathbb C^\ell\\
%x\otimes e_1&\mapsto\sum_{i=1}^{\ell}K_ix  \otimes e_i
%\end{align*}
%and notice that $U_0^*U_0$ equals the identity on $\mathbb C^n\otimes \mathbb C e_1\cong\mathbb C^n$ (so $U_0$ is an isometry).
%In other words the ``columns'' of $U_0$ form an orthonormal system which, 
%subsequently, can be completed to an orthonormal basis in $\mathbb C^{n}\otimes\mathbb C^{\ell}$. By extending the ``columns'' of $U_0$ with these additional vectors we obtain a unitary matrix $U\in\mathbb C^{n\ell\times n\ell} $ which does the job.
because the Kraus operators satisfy $\sum_{i=1}^\ell K_i^*K_i=\mathbbm 1$ the ``columns'' of $U_0$ form an orthonormal system which can then be extended to an orthonormal basis.
This results in a unitary $U$ with the desired property,
Thm.~6.18 \& Eq.~(6.23) in \cite{Holevo12}, cf.~also Appendix~A in \cite{vE22_Stinespring}.
Note that one can also obtain Eq.~\eqref{eq:intro_0} directly from Stinespring's
original dilation theorem for $C^*$-algebras \cite{Stinespring55} (as done in, e.g., Thm.~4.18 ff. in \cite{Heinosaari12}) which, however, is less explicit in the sense that it translates an abstract object (the isometry $V$) into another one (the unitary $U$).
Hence our aim is to start from the Kraus operators as they are more relevant from
an application point-of-view\footnote{
However, be aware that the existence of Kraus operators---in particular in infinite dimensions---is usually proven via Stinespring's dilation theorem for $C^*$-algebras, cf.~Ch.~9, Thm.~2.3 in \cite{Davies76}.
}.

When trying to generalize the idea stated previously to infinite dimensions one immediately runs into two problems:
\begin{itemize}
\item For general spaces, not every isometry can be extended to a unitary. The prime example is the left-shift on $\ell^2(\mathbb N)$ truncated to a subspace:
\begin{align*}
\sigma_{l,0}:\{x\in\ell^2(\mathbb N)\,:\,x_1=0\}&\to\ell^2(\mathbb N)\\
(0,x_1,x_2,x_3,\ldots)&\mapsto(x_1,x_2,x_3,\ldots)
\end{align*}
This is an isometry which cannot be extended to a unitary (or even to an isometry) on $\ell^2(\mathbb N)$;
the reason for this is that $\sigma_{l,0}$ it is already surjective so there is ``no room left''
in the co-domain $\ell^2(\mathbb N)$.
Of course, this is a purely infinite-dimensional effect because in finite 
dimensions every isometric endomorphism is automatically unitary by the rank-nullity theorem.
For more on the extension of isometric operators defined only on a subspace, as well as the connection to extending symmetric operators refer to p.~387~ff.~in \cite{Ludwig83}.
\item Even when assuming that domain and range of the isometry $U_0$ have the same co-dimension,
if this number is infinite, then the above strategy of generating an orthonormal basis relies on the axiom of choice.
This feels counter-intuitive because something as fundamental as describing physical processes via environment interactions should work in any ``reasonable'' mathematical framework.
\end{itemize}
In their original proof of Eq.~\eqref{eq:intro_0} for infinite-dimensional $\mathcal H$ Hellwig \& Kraus 
\cite{HK70} (resp.~\cite{Kraus71,Kraus73,Kraus83})
resolved this problem
by treating $\mathcal H\otimes\mathcal K$ as ``$\operatorname{dim}\mathcal K$-many copies'' of $\mathcal H$ and then simply adding another $\mathcal H$ to it.
As it turns out this is all the extra space one needs in the co-domain to guarantee that an isometry can be extended to something unitary.
For the reader's convenience we summarize their argument in Appendix~\ref{app_A}.

In contrast, in this article we want to stay closer to the proof idea from finite dimensions
(i.e.~defining an isometry on an explicit subspace of $\mathcal H\otimes\mathcal K$  via Kraus operators).
For this we need a result on extensions of isometries on a linear subspace, which
is what the next section will be about.
Then in Section~\ref{sec_main} we present an alternative construction of the Stinespring form of arbitrary quantum maps; in
the spirit of George P{\'o}lya
\textit{``Two proofs are better than one. ``It is safe riding at
two anchors.''\,''} (p.~62 in \cite{Polya04}).
Finally we compare the three methods for constructing a Stinespring form described in this paper by means of a simple example in Section \ref{sec_comparison_example}.

\section{Preliminaries: Extensions of Isometries}\label{sec_prelim}
%In what follows all Hilbert spaces will be complex, although we remark that the subsequent results also hold for real Hilbert spaces if the auxiliary spaces $\mathbb C^2$ are replaced by $\mathbb R^2$.
Sz.-Nagy’s dilation theorem is undoubtedly
among the most famous results in operator theory as a whole,
and in the branch of operator dilations in particular.
It states that for every contraction $T$ on a Hilbert space $\mathcal H$ (i.e.~$\|T\|_{\mathsf{op}}\leq 1$ with $\|\cdot\|_{\mathsf{op}}$ being the usual operator norm) the operator
\begin{equation}\label{eq:sznagy}
\begin{pmatrix}
T&\sqrt{\mathbbm 1-TT^*}\\
\sqrt{\mathbbm1-T^*T}&-T^*
\end{pmatrix},
\end{equation}
on $\mathcal H\times\mathcal H$---where Eq.~\eqref{eq:sznagy} is clearly a dilation of $T$---is unitary (Appendix, Sec.~4 in \cite{SzNagy90}, cf.~also \footnote{
Actually in \cite{SzNagy90} it is proven that $\footnotesize\begin{pmatrix}
T&\sqrt{\mathbbm 1-TT^*}\\
-\sqrt{\mathbbm1-T^*T}&T^*
\end{pmatrix}$ is a unitary dilation of $T$. However, multiplying this from the left with the unitary $\mathbbm1\oplus(-\mathbbm1)$ yields Eq.~\eqref{eq:sznagy}.
}
and Ch.~VI in \cite{Foias90}).

Now assume that we are dealing with an isometry $V_0$ defined on a (closed) subspace $\mathcal M\subseteq\mathcal H$.
One way to extend $V_0$ to an operator on all of $\mathcal H$ is via
$V':=V_0P_{\mathcal M}$ where $P_{\mathcal M}$ be the orthogonal projection\footnote{
Because $\mathcal M$ is closed by assumption, it is a Hilbert space itself (Example~11.3 in \cite{MeiseVogt97en}) which implies that $P_{\mathcal M}$ as well as $V_0^*$ are well-defined (Sec.~2.5 \& Thm.~2.4.2 in \cite{Kadison83}). 
}
onto $\mathcal M$.
Note that $V'$ is a contraction ($\|V'\|_{\mathsf{op}}\leq\|V_0\|_{\mathsf{op}}\|P_{\mathcal M}\|_{\mathsf{op}}=1$)
so we can apply Sz.-Nagy’s dilation theorem to it.
For this is it crucial to observe that $(V')^*$ coincides with the adjoint $V_0^*:\mathcal H\to\mathcal M$ of $V_0$ (i.e.~$\langle V_0x,y\rangle=\langle x,V_0^*y\rangle$ for all $x\in\mathcal M$, $y\in\mathcal H$) because
\begin{align*}
\langle x,(V')^*y\rangle=\langle V'x,y\rangle=\langle V_0(P_{\mathcal M}x),y\rangle&=\langle P_{\mathcal M}x,V_0^*y\rangle\\
&=\langle P_{\mathcal M}x,V_0^*y\rangle+\langle(\mathbbm1- P_{\mathcal M})x,V_0^*y\rangle=\langle x,V_0^*y\rangle
\end{align*}
for all $x,y\in\mathcal H$. In the second-to-last step we used that $V_0^*y\in\mathcal M$ but $(\mathbbm 1-P_{\mathcal M})x\in\mathcal M^\perp$ so the inner product of the two vanishes.
With this as well as the fact that $V_0$ is an isometry (i.e.~$V_0^*V_0=\mathbbm1$) Eq.~\eqref{eq:sznagy} becomes
\begin{align*}
\begin{pmatrix}
V'&\sqrt{\mathbbm1-V'(V')^*}\\
\sqrt{\mathbbm1-(V')^*V'}&-(V')^*
\end{pmatrix}&=
\begin{pmatrix}
V_0P_{\mathcal M}&\sqrt{\mathbbm1-V_0(P_{\mathcal M}V_0^*)}\\
\sqrt{\mathbbm1-V_0^*V_0P_{\mathcal M}}&-V_0^*
\end{pmatrix}\\
&=\begin{pmatrix}
V_0P_{\mathcal M}&\sqrt{\mathbbm1-V_0V_0^*}\\
\sqrt{\mathbbm1-P_{\mathcal M}}&-V_0^*
\end{pmatrix}.
\end{align*}
Finally, observe that $\mathbbm1-V_0V_0^*,\mathbbm1-P_{\mathcal M}$ are 
orthogonal projections so in particular they are positive semi-definite (note $\langle x,Px\rangle=\langle x,P^2x\rangle=\langle x,P^*Px\rangle=\langle Px,Px\rangle=\|Px\|^2\geq 0$ for all $P^*=P=P^2$)
and thus they are their own square root.
This yields the unitary dilation
\begin{equation}\label{eq:dilation_V0}
\begin{pmatrix}
V_0P_{\mathcal M}&\mathbbm1-V_0V_0^*\\
\mathbbm1-P_{\mathcal M}&-V_0^*
\end{pmatrix}.
\end{equation}
on $\mathcal H\times\mathcal H$ of $V_0$.
However, the direct product of vector spaces is not as meaningful in a ``quantum physics setting''
because there multipartite systems are described via the tensor product of the respective Hilbert spaces.
This is why we reformulate our preceding calculations as follows:
\begin{lemma}\label{lemma_iso_1}
Let $\mathcal M\subseteq\mathcal H$ be a closed subspace of a complex Hilbert 
space $\mathcal H$, and let $V_0:\mathcal M\to\mathcal H$ be an isometry.
%The following statements hold:
%\begin{itemize}
%\item[(i)]
There 
exists $U:\mathcal H\otimes\mathbb C^2\to\mathcal H\otimes\mathbb C^2$ unitary which extends $V_0$ in the sense that $U(x\otimes e_1)=V_0x\otimes e_1$ for all $x\in\mathcal M$.
%\item[(ii)] There 
%exists $U':\mathcal H\otimes\mathbb C^2\to\mathcal H\otimes\mathbb C^2$
%self-adjoint and unitary (i.e.~an involution) which extends $V_0$ in the sense that $U(x\otimes e_2)=V_0x\otimes e_1$ for all $x\in\mathcal M$.
%\end{itemize}
\end{lemma}
\begin{proof}
In Eq.~\eqref{eq:dilation_V0} we already saw that
\begin{align*}
U_0:\mathcal H\times\mathcal H&\to\mathcal H\times\mathcal H\\
\begin{pmatrix}
x\\y
\end{pmatrix}&\mapsto\begin{pmatrix}
V_0P_{\mathcal M}x+(\mathbbm1-V_0V_0^*)y\\
(\mathbbm1-P_{\mathcal M})x-V_0^*y
\end{pmatrix}
\end{align*}
is a unitary dilation of $V_0$.
All that is left to do is to ``translate'' $U_0$ into a unitary on $\mathcal H\otimes\mathbb C^2$. For this note that the space $\mathcal H\otimes\mathbb C^2$ is isometrically isomorphic to $\mathcal H\times\mathcal H$ by means of the map $J:\mathcal H\times\mathcal H\to\mathcal H\otimes\mathbb C^2$, $(x,y)\mapsto x\otimes e_1+y\otimes e_2$ (i.e.~$J$ is a unitary transformation, Remark~2.6.8 in \cite{Kadison83}).
Thus we define $U:\mathcal H\otimes\mathbb C^2\to\mathcal H\otimes\mathbb C^2$, $x\mapsto (J\circ U_0\circ J^{-1})(x)$ as visualized in the following commutative diagram:
$$
\begin{tikzcd}
\mathcal H\otimes\mathbb C^2 \arrow[r, "U"] \arrow[d, "J^{-1}"] & \mathcal H\otimes\mathbb C^2              \\
\mathcal H\times\mathcal H \arrow[r, "U_0"]                     & \mathcal H\times\mathcal H \arrow[u, "J"]
\end{tikzcd}
$$
Obviously $U$ is unitary as it is a composition of unitaries, and $U$ satisfies the desired extension property because for all $x\in\mathcal M$
\begin{align*}
U(x\otimes e_1)=( J\circ U_0\circ J^{-1} )(x\otimes e_1)&=(J\circ U_0)\begin{pmatrix}
x\\0
\end{pmatrix}\\
&=J\begin{pmatrix}
V_0P_{\mathcal M}x\\
(\mathbbm1-P_{\mathcal M})x
\end{pmatrix}=J\begin{pmatrix}
V_0x\\
x-x
\end{pmatrix}=V_0x\otimes e_1\,.\qedhere
\end{align*}
\end{proof}

With this lemma at our disposal we are ready to have a different look at (Stinespring) dilations of quantum maps.
\section{The Stinespring Form of Quantum Maps}\label{sec_main}
Given any complex Hilbert spaces $\mathcal H,\mathcal{K}$, 
recall that a \textit{(Schrödinger) quantum map} is a linear map $\Phi$ between trace classes (Ch.~16 in \cite{MeiseVogt97en}) $\mathcal B^1(\mathcal H)$, $\mathcal B^1(\mathcal{K})$ which preserves the trace and is completely positive, that is,
for all $n\in\mathbb N$ the extended map $\Phi\otimes\operatorname{id}_n:\mathcal B^1(\mathcal H\otimes\mathbb C^n)\to\mathcal B^1(\mathcal{K}\otimes\mathbb C^n)$ sends positive semi-definite operators to positive semi-definite operators. Equivalently, $\Phi$ is completely positive if and only if there exists a family $(K_j)_{j\in J}\subset\mathcal B(\mathcal H,\mathcal{K})$---called \textit{Kraus operators}---such that 
$
\Phi(\cdot)=\sum_{j\in J} K_j(\cdot)K_j^*
$, where the sum converges in trace norm and $\sum_{j\in J}K_j^*K_j$ converges strongly to a bounded operator, cf.~%
%\cite[Thm.~1]{Kraus83} or 
Ch.~9, Thm.~2.3 in \cite{Davies76}.
If $\Phi$ is additionally trace-preserving, then $\sum_{j\in J}K_j^*K_j$ strongly converges to the identity.
The collection of all completely positive and trace-preserving maps (called \textsc{cptp} or quantum maps) from $\mathcal H$ to $\mathcal{K}$ will be denoted by $\textsc{cptp}(\mathcal H,\mathcal{K})$ (and $\textsc{cptp}(\mathcal H):=\textsc{cptp}(\mathcal H,\mathcal H)$).
A particularly important element in $\textsc{cptp}(\mathcal H\otimes\mathcal{K},\mathcal H)$ is the partial trace $\operatorname{tr}_{\mathcal K}$ which is the unique
linear map satisfying $\operatorname{tr}(\operatorname{tr}_{\mathcal K}(A)B)=\operatorname{tr}(A(B\otimes\mathbbm1))$ for all $A\in\mathcal B^1(\mathcal H\otimes\mathcal K)$, $B\in\mathcal B(\mathcal H)$, cf.~Def.~2.68 ff.~in \cite{Heinosaari12}.
Finally, $\mathbb D(\mathcal H)$ will denote the set of all \textit{states}, i.e.~all positive semi-definite trace-class operators of trace one.
With all the notation in place let us state and prove our main result:

\begin{theorem}\label{thm1}
Given any $\Phi\in \textsc{cptp}(\mathcal H)$ there exists a Hilbert space $\mathcal K$, a unit vector $\psi\in\mathcal K$, and a unitary operator $U$ on $\mathcal H\otimes \mathcal K$ such that
$$
\Phi\equiv\operatorname{tr}_{\mathcal K}\big(U((\cdot)\otimes|\psi\rangle\langle\psi)U^*\big)\,.
$$
Moreover,
\begin{itemize}
\item[(i)] if $(K_j)_{j\in J}$ is any set of Kraus operators of $\Phi$, then one can choose $\mathcal K$ to be $\ell^2(J)\otimes\mathbb C^2$ and $\psi:=e_{j_0}\otimes e_1$ for any $j_0\in J$.
In particular if $\mathcal H$ is separable, then $\mathcal K$ can be chosen separable, as well.
\item[(ii)] $U$ can be chosen such that the auxiliary qubit is a catalyst, i.e.~for all $\rho\in\mathbb D(\mathcal H)$ there exists $\omega\in\mathbb D(\mathcal H\otimes\ell^2(J))$ such that
$
U(\rho\otimes|e_{j_0}\rangle\langle e_{j_0}| \otimes |e_1\rangle\langle e_1|)U^*=\omega\otimes|e_1\rangle\langle e_1|
$.
\end{itemize}
\end{theorem}
\begin{proof}
Starting from any set of Kraus operators $(K_j)_{j\in J}$ for $\Phi$ as well as an arbitrary
(but fixed) $j_0\in J$ we define a map on a (closed) subspace of $\mathcal H\otimes\ell^2(J)$ via
\begin{align*}
V_0:\mathcal H\otimes \mathbb Ce_{j_0} = \{x\otimes e_{j_0}:x\in\mathcal H\}&\to \mathcal H\otimes\ell^2(J)\\
x\otimes e_{j_0}&\mapsto\sum_{j\in J}K_jx\otimes e_j\,.
\end{align*}
Here $e_j\in\ell^2(J)$ is the ``$j$-th standard basis vector'' $e_j:J\to\mathbb C$, $j'\mapsto\delta_{jj'}$.
Because $(K_jx\otimes e_j)_{j\in J}$ is an orthogonal set in the Hilbert space  $\mathcal H\otimes\ell^2(J)$, by Prop.~2.2.5 in \cite{Kadison83} $\sum_{j\in J}K_jx\otimes e_j$ exists if and only if $\sum_{j\in J}\|K_jx\|^2<\infty$; but this holds due to
\begin{equation}\label{eq:4'}
\sum_{j\in J_F}\|K_jx\|^2=\sum_{j\in J_F}\langle K_jx,K_jx\rangle=\Big\langle x,\sum_{j\in J_F}K_j^*K_jx\Big\rangle
\end{equation}
for all finite subsets $J_F\subseteq J$ and all $x\in\mathcal H$
together with the fact that the Kraus operators satisfy $\sum_{j\in J}K_j^*K_j\to\mathbbm 1$ in the strong (hence the weak) operator topology. This shows that $V_0$ is well-defined.
Moreover, $V_0$ is an isometry because, again, Prop.~2.2.5 from \cite{Kadison83} for all $x\in\mathcal H$ yields
$$
\|V_0(x\otimes e_{j_0})\|^2=\Big\|\sum_{j\in J}K_jx\otimes e_j\Big\|^2=\sum_{j\in J}\|K_jx\|^2\overset{\eqref{eq:4'}}=\|x\|^2\,.
$$
Now Lemma~\ref{lemma_iso_1} comes into play:
it lets us extend $V_0$ to a unitary $U$ on $\mathcal H\otimes\ell^2(J)\otimes\mathbb C^2$
which by Eq.~\eqref{eq:dilation_V0}---when identifying $\mathcal H\otimes\ell^2(J)\otimes\mathbb C^2\cong (\mathcal H\otimes\ell^2(J))\times(\mathcal H\otimes\ell^2(J))$---is of the form
\begin{equation}\label{eq:U_block_form}
\begin{pmatrix}
V_0(\mathbbm1\otimes|e_{j_0}\rangle\langle e_{j_0}|)&\mathbbm1-V_0V_0^*\\
\mathbbm 1\otimes(\mathbbm1-|e_{j_0}\rangle\langle e_{j_0}|)&-V_0^*
\end{pmatrix}=\begin{pmatrix}
\sum_{j\in J}K_j\otimes|e_j\rangle\langle e_{j_0}|&\mathbbm1-(\sum_{j,j'\in J}K_jK_{j'}^*\otimes|e_j\rangle\langle e_{j'}|)\\
\mathbbm 1\otimes(\mathbbm1-|e_{j_0}\rangle\langle e_{j_0}|)&-\sum_{j\in J}K_j^*\otimes|e_{j_0}\rangle\langle e_j|
\end{pmatrix}.
\end{equation}
%Obviously $U(z\otimes e_1\otimes e_1)=V_0z$ for all $z\in\mathcal H\otimes \mathbb Ce_{j_0}$.
Define $\psi:=e_{j_0}\otimes e_1\in\ell^2(J)\otimes\mathbb C^2=:\mathcal K$.
Note that if $\mathcal H$ is separable, then $J$ can be chosen countable (Ch.~9, Thm.~2.3 in \cite{Davies76}) meaning $\mathcal K$ is separable (this proves (i)).
Given $x,y\in\mathcal H$ we compute
\begin{align*}
\operatorname{tr}_{\mathcal K}\big(U(|x\rangle\langle y|\otimes|\psi\rangle\langle\psi|)U^*\big)&=\operatorname{tr}_{\ell^2(J)\otimes\mathbb C^2}\big(|U(x\otimes e_{j_0}\otimes e_1)\rangle\langle U(y\otimes e_{j_0}\otimes e_1)|\big)\\
&=\operatorname{tr}_{\ell^2(J)\otimes\mathbb C^2}\big(|V_0(x\otimes e_{j_0})\rangle\langle V_0(y\otimes e_{j_0})|\otimes|e_1\rangle\langle e_1|\big)\\
&=\operatorname{tr}_{\ell^2(J)}\Big(\Big|\sum_{j\in J}K_jx\otimes  e_j\Big\rangle\Big\langle \sum_{j'\in J}K_{j'}y\otimes e_{j'}\Big|\Big)\\
&=\sum_{j,j'\in J}\operatorname{tr}_{\ell^2(J)}\big(|K_jx\rangle\langle K_{j'}y| \otimes  |e_j\rangle\langle e_{j'}|\big)\\
&=\sum_{j,j'\in J}K_j|x\rangle\langle y|K_{j'}^*\langle e_{j'},e_j\rangle=\sum_{j\in J}K_j|x\rangle\langle y|K_j^*=\Phi(|x\rangle\langle y|)\,.
\end{align*}
In the second-to-last line we used that the partial trace (just like every quantum map) is continuous (Prop.~2 in \cite{vE_dirr_semigroups}).
Hence this also holds for both $\Phi$ and $\operatorname{tr}_{\mathcal K}(U((\cdot)\otimes|\psi\rangle\langle\psi)U^*)$, meaning they coincide on all of $\mathcal B^1(\mathcal H)$
because we showed that they coincide on the dense subset $\operatorname{span}\{|x\rangle\langle y|:x,y\in\mathcal H\}\subseteq \mathcal B^1(\mathcal H)$.
The only statement left to prove is the catalyst property (ii); this follows readily from
the extension property in Lemma~\ref{lemma_iso_1}, resp.~from Eq.~\eqref{eq:U_block_form}.
\end{proof}
\noindent 
Note that this construction works for any set of Kraus operators so the smallest auxiliary system one can get this way has dimension ``Kraus rank'' times two.
In particular, when disregarding the catalyst qubit this yields a smaller auxiliary system than the construction of Hellwig and Kraus,
and the dimension being the Kraus rank is tight for general systems.
However, in practice one could also choose a ``non-minimal'' set of Kraus operators -- if, e.g., the corresponding unitary is easier to implement in practice --
or use a different (not fully general) construction altogether.

\begin{remark}
Be aware that the Stinespring form
\begin{itemize}
\item[(i)] 
is not limited to quantum maps with same domain and co-domain:
given $\Phi\in\textsc{cptp}(\mathcal H,\mathcal{K})$ one
obtains an extension of Theorem~\ref{thm1} to arbitrary quantum maps
via the auxilliary map
$
\rho\mapsto |\psi'\rangle\langle\psi'|\otimes\Phi(\operatorname{tr}_{\mathcal{K}}(\rho))\in\textsc{cptp}(\mathcal H\otimes\mathcal{K})
$
where $\psi'\in\mathcal H$ is any unit vector,
cf.~also Coro.~1 in \cite{vE_dirr_semigroups}.
\item[(ii)] exists equivalently in the Heisenberg picture:
given any (Heisenberg) quantum channel $\Phi^*$ -- meaning $\Phi^*$ is any completely positive, unital (i.e.~identity preserving), and ultraweakly continuous map --
one has $\Phi^*=\operatorname{tr}_{|\psi\rangle\langle\psi|}(U^*((\cdot)\otimes\mathbbm1)U^*)$ where
 $\psi,U$ are the same as in Thm.~\ref{thm1}
and $\operatorname{tr}_{|\psi\rangle\langle\psi|}$ is the partial trace w.r.t.~the state $|\psi\rangle\langle\psi|$ (Ch.~9, Lemma~1.1 in \cite{Davies76}).
For more on the Stinespring form in the Heisenberg picture and its relation to Stinespring's theorem for $C^*$-algebras we refer to Coro.~2 ff.~in \cite{vE_dirr_semigroups}.
\end{itemize}
\end{remark}
\section{Comparing Constructions: An Example}\label{sec_comparison_example}
To better understand the different techniques for generating Stinespring forms 
%from a set of Kraus operators 
let us illustrate and compare them via a simple example. While the following example is finite-dimensional -- hence the construction described in the introduction yields the smallest auxiliary system -- such a choice will clarify how the two constructions intended for infinite dimensions circumvent the problems discussed in the introduction.

Let $\Phi\in\textsc{cptp}(n)$ be given such that $\{K_1,K_2\}\subset\mathbb C^{n\times n}$ is a set of Kraus operators for $\Phi$.
As explained in the introduction, the standard construction first collects $K_1,K_2$ in a $2n\times 2n$-matrix:
%($2n$ because the map $\Phi$ is defined on an $n$-dimensional space and has $2$ Kraus operators)
\begin{equation}\label{eq:ex_1}
\begin{pmatrix}
K_1&0\\K_2&0
\end{pmatrix}
\end{equation}
Because $K_1^*K_1+K_2^*K_2=\mathbbm1$ the first $n$ columns in~\eqref{eq:ex_1} form an orthonormal system in $\mathbb C^{2n}$ which can be extended to an orthonormal basis. Build $U_{12},U_{22}\in\mathbb C^{n\times n}$ from the added vectors such that
$$
U_{\mathrm{F}}:=\begin{pmatrix}
K_1&U_{12}\\K_2&U_{22}
\end{pmatrix}
$$
is unitary, and one has $\Phi\equiv\operatorname{tr}_{\mathbb C^2}(U_{\mathrm{F}}((\cdot)\otimes|e_1\rangle\langle e_1|)U_{\mathrm{F}}^*)$.

Next, let us carry out the construction of Hellweg \& Kraus (cf.~Appendix~\ref{app_A}): 
similar to Eq.~\eqref{eq:ex_1} one starts with the isometry
\begin{equation}\label{eq:ex_2}
A:=\begin{pmatrix}
K_1\\K_2
\end{pmatrix}\in\mathbb C^{2n\times n}\,.
\end{equation}
However, instead of completing it to a unitary by adding suitable columns -- which cannot be guaranteed in infinite dimensions -- they enlarge the auxiliary system in order to define
$$
U_{\mathrm{K}}:=\begin{pmatrix}
0&A^*\\A&-(\mathbbm 1-AA^*)
\end{pmatrix}=\begin{pmatrix}
0&K_1^*&K_2^*\\
K_1&K_1K_1^*-\mathbbm1&K_1K_2^*\\
K_2&K_2K_1^*&K_2K_2^*-\mathbbm1
\end{pmatrix}.
$$
This $U_{\mathrm{K}}$ satisfies $\Phi\equiv\operatorname{tr}_{\mathbb C^3}(U_{\mathrm{K}}((\cdot)\otimes|e_1\rangle\langle e_1|)U_{\mathrm{K}}^*)$.

Finally, the construction presented in our paper also starts with Eq.~\eqref{eq:ex_2} but then turns $A$ into the ``block matrix'' from Eq.~\eqref{eq:ex_1} 
%by means of a projection onto the ``first column'' of this matrix.
as part of the larger matrix
$$
U_{\mathrm{N}}:=\begin{pmatrix}
K_1&0&\mathbbm1-K_1K_1^*&-K_1K_2^*\\
K_2&0&-K_2K_1^*&\mathbbm1-K_2K_2^*\\
0&0&-K_1^*&-K_2^*\\
0&\mathbbm1&0&0
\end{pmatrix}.
$$
Choosing the auxiliary system to be $\mathbb C^2\otimes\mathbb C^2$ this $U_N{\mathrm{N}}$
satisfies $\Phi\equiv\operatorname{tr}_{\mathbb C^2\otimes\mathbb C^2}(U_{\mathrm{N}}((\cdot)\otimes|e_1\rangle\langle e_1|)U_{\mathrm{N}}^*)$.
Note that a shuffled version of $U_{\mathrm{K}}$ ``appears'' in
$U_{\mathrm{N}}$ (up to some minus signs which can be neglected) 
and the two matrices differ by an identity on a complementary subspace.
Moreover, the fact that the ``effective dimension'' 
of our construction is smaller (in the sense that part of the space $U_{\mathrm{N}}$ acts on is a catalyst w.r.t.~$U_{\mathrm{N}}$) 
also manifests here: the first column of the upper left block of $U_{\mathrm{N}}$ features only the Kraus operators.
This is why
$$
U_{\mathrm{N}}(\rho\otimes|e_1\rangle\langle e_1|)U_{\mathrm{N}}^*=
\begin{pmatrix}
K_1\rho K_1^*&K_1\rho K_2^*\\
K_2\rho K_1^*&K_2\rho K_2^*
\end{pmatrix}\otimes |e_1\rangle\langle e_1|=
U_{\mathrm{F}}(\rho\otimes|e_1\rangle\langle e_1|)U_{\mathrm{F}}^*\otimes|e_1\rangle\langle e_1|\,,
$$
whereas
$$
U_{\mathrm{N}}(\rho\otimes|e_1\rangle\langle e_1|)U_{\mathrm{N}}^*= \begin{pmatrix}
0&0&0\\
0&K_1\rho K_1^*&K_1\rho K_2^*\\
0&K_2\rho K_1^*&K_2\rho K_2^*
\end{pmatrix}=0\oplus U_{\mathrm{F}}(\rho\otimes|e_1\rangle\langle e_1|)U_{\mathrm{F}}^*\,.
$$
As a follow-up question one could ask how these different constructions manifest in the 
(practically important) Gaussian setting mentioned in the introduction.
While this is not obvious -- as there it is not the Kraus rank that matters, but rather the number of modes --
pursuing this question could lead to a deeper understanding of how the Kraus
operators are connected to seemingly unrelated properties of a channel.

%While one could re-shuffle $U_{\mathrm{K}}$, doing so means one has to re-shuffle the environment state accordingly which always yields a $1$ precisely where the zero is located in $U_{\mathrm{K}}$.
%In other words the Kraus operators can never be in a diagonal sub-block of $U_{\mathrm{K}}$ this way.
\begin{acknowledgments}
I would like to thank Jens Eisert for useful comments during the preparation of this manuscript. In particular he pointed out
%%and the anonymous referee
%for valuable and constructive comments during the preparation of this manuscript.
some references on Stinespring forms of bosonic Gaussian channels which I was not aware of yet.
%Moreover I am grateful to the anonymous referee for their valuable comments which led to an improved presentation of the material.
%This research is part of the Bavarian excellence network \textsc{enb}
%via the International PhD Programme of Excellence
%\textit{Exploring Quantum Matter} (\textsc{exqm}), as well as the \textit{Munich Quantum Valley} of the Bavarian
%State Government with funds from Hightech Agenda \textit{Bayern Plus}.
This work has been supported by the Einstein Foundation (Einstein Research Unit on Quantum Devices) and the MATH+ Cluster of Excellence.
\end{acknowledgments}

\appendix
\section{Original Proof of Hellwig and Kraus}\label{app_A}
This appendix will revolve around the following statement, respectively the proof given by Hellwig and Kraus (originally in \cite{HK70}, and in more detail in Sec.~4 in \cite{Kraus73} or Thm.~2 in \cite{Kraus83}):\medskip

\setlength{\leftskip}{1cm}\setlength{\rightskip}{1cm}

{\it\noindent Given any complex Hilbert space $\mathcal H$ and any quantum map $\Phi$ on $\mathcal H$ there exists a Hilbert space $\mathcal K$, a unit vector $\psi\in\mathcal K$, and a \textnormal{self-adjoint} unitary operator $U$ on $\mathcal H\otimes \mathcal K$ such that
$$
\Phi\equiv\operatorname{tr}_{\mathcal K}\big(U((\cdot)\otimes|\psi\rangle\langle\psi)U^*\big)\,.
$$
If $(K_j)_{j\in J}$ is a set of Kraus operators of $\Phi$, then one can choose $\mathcal K$ to be $\ell^2(J\cup\{\mathrm{s}\})$ where $\mathrm{s}$ is any symbol not in $J$.}\medskip

\setlength{\leftskip}{0pt}\setlength{\rightskip}{0pt}

\noindent We note that their proof was given for separable Hilbert spaces $\mathcal H$ but extends without further ado to arbitrary Hilbert spaces. Their construction goes as follows.
Starting from a set of Kraus operators $(K_j)_{j\in J}$ for $\Phi$ (cf.~Ch.~9, Thm.~2.3 in \cite{Davies76}) one first defines the following objects:
\begin{itemize}
\item $J_s:=J\cup\{s\}$ where $s$ is any symbol not in $J$
\item $\mathcal K:=\ell^2(J_s)$ is the Hilbert space of all functions $f:J_s\to\mathbb C$ which are square-summable, i.e.~$\sum_{j\in J_s}|f(j)|^2<\infty$, cf.~Example~1.7.3 \& Example~2.1.12 in \cite{Kadison83}
\item $\mathcal H_s:=\mathcal H=:\mathcal H_j$ for all $j\in J$
\item $\iota:\mathcal H_s\oplus\bigoplus_{j\in J}\mathcal H_j\to \mathcal H\otimes\mathcal K$ is the isometric isomorphism defined via $x_s\oplus\bigoplus_{j\in J}x_j\mapsto \sum_{j\in J_s}x_j\otimes e_j$. Hence $\mathcal H_s\oplus\bigoplus_{j\in J}\mathcal H_j\cong  \mathcal H\otimes\mathcal K$, cf.~Remark~2.6.8 in \cite{Kadison83}.
\end{itemize}
The idea now is to define an isometry $A:\mathcal H\to\bigoplus_{j\in J}\mathcal H_j$, embed it into a unitary operator $U_0$ on $\mathcal H_s\oplus\bigoplus_{j\in J}\mathcal H_j$, and finally use $\iota$ to translate $U_0$ into a unitary operator $U$ on $\mathcal H\otimes\mathcal K$ as visualized in the following diagram:
$$
\begin{tikzcd}
\mathcal H\otimes\ell^2(J_s) \arrow[r, "U"] \arrow[d, "\iota^{-1}"] & \mathcal H\otimes\ell^2(J_s)                                        \\
\mathcal H_s\oplus\bigoplus_{j\in J}\mathcal H_j \arrow[r, "U_0"]   & \mathcal H_s\oplus\bigoplus_{j\in J}\mathcal H_j \arrow[u, "\iota"]
\end{tikzcd}
$$
They started by defining $A:\mathcal H_s\to\bigoplus_{j\in J}\mathcal H_j$ via $Ax:=\bigoplus_{j\in J}K_jx$. One readily verifies that $A$ is an isometry (so in particular well-defined) because
$$
\|Ax\|^2=\sum_{j\in J}\|K_jx\|^2=\sum_{j\in J}\langle x,K_j^*K_jx\rangle=\|x\|^2
$$
as $\sum_{j\in J}K_j^*K_j\to\mathbbm1$ in the strong operator topology. 
With this they defined $U_0$ via\footnote{
Originally, Hellwig and Kraus considered general 
quantum operations, cf.~Sec.~4 in \cite{Heinosaari12} meaning the Kraus operators 
only need to satisfy $\sum_{j\in J}K_j^*K_j\leq\mathbbm1$. This made their construction of $U_0$ a bit more involved (Eq.~(4.3) in \cite{Kraus73}); however, as we are only interested in a dilation of quantum maps here we may use the simpler version of $U_0$ (Eq.~(5.27) in \cite{Kraus83}).
}
\begin{align*}
U_0:\mathcal H_s\oplus\bigoplus_{j\in J}\mathcal H_j&\to\mathcal H_s\oplus\bigoplus_{j\in J}\mathcal H_j\\
\begin{pmatrix}
x\\y
\end{pmatrix}&\mapsto\begin{pmatrix}
A^*y\\Ax-(\mathbbm1-AA^*)y
\end{pmatrix}=\begin{pmatrix}
0&A^*\\A&-(\mathbbm 1-AA^*)
\end{pmatrix}\begin{pmatrix}
x\\y
\end{pmatrix}.
\end{align*}
Evidently, $U_0$ is a self-adjoint involution; hence $U_0$ is unitary and so is the ``translated'' operator $U:=\iota\circ U_0\circ\iota^{-1}$ on $\mathcal H\otimes\mathcal K$ (because $\iota$ is a unitary transformation).
Note that
\begin{equation}\label{eq:app_1}
U(x\otimes e_s)=(\iota\circ U_0)\begin{pmatrix}
x\\0
\end{pmatrix}=\iota\begin{pmatrix}
0\\Ax
\end{pmatrix}= \iota\begin{pmatrix}
0\\\bigoplus_{j\in J}K_jx
\end{pmatrix} =\sum_{j\in J}K_jx\otimes e_j
\end{equation}
for all $x\in\mathcal H_s=\mathcal H$.
Defining $\psi:=e_s\in\ell^2(J_s)$ one for all $x,y\in\mathcal H$ finds
\begin{align*}
\operatorname{tr}_{\ell^2(J_s)}\big(U(|x\rangle\langle y|\otimes|e_s\rangle\langle e_s|)U^*\big)&\overset{\hphantom{\eqref{eq:app_1}}}=\operatorname{tr}_{\ell^2(J_s)}\big(|U(x\otimes e_s)\rangle\langle U(y\otimes e_s)|\big)\\
&\overset{\eqref{eq:app_1}}=\operatorname{tr}_{\ell^2(J)}\Big(\Big|\sum_{j\in J}K_jx\otimes  e_j\Big\rangle\Big\langle \sum_{j'\in J}K_{j'}x\otimes e_{j'}\Big|\Big)\\
&\overset{\hphantom{\eqref{eq:app_1}}}=\sum_{j,j'\in J}\operatorname{tr}_{\ell^2(J)}\big(|K_jx\rangle\langle K_{j'}x| \otimes  |e_j\rangle\langle e_{j'}|\big)\\
&\overset{\hphantom{\eqref{eq:app_1}}}=\sum_{j,j'\in J}K_j|x\rangle\langle y|K_{j'}^*\langle e_{j'},e_j\rangle=\sum_{j\in J}K_j|x\rangle\langle y|K_j^*=\Phi(|x\rangle\langle y|)\,.
\end{align*}
A standard continuity argument shows that $\Phi\equiv\operatorname{tr}_{\mathcal K}(U((\cdot)\otimes|\psi\rangle\langle\psi)U^*)$ on all of $\mathcal B^1(\mathcal H)$.

%\begin{remark}
%In \cite{Kraus71}: construction when symmetry constraints enter the picture
%\end{remark}

\bibliography{../../../../control21vJan20.bib}
\end{document}